\documentclass[envcountsect,envcountsame,runningheads]{llncs}
\usepackage{graphics}

\newcommand{\BWTs}{\ensuremath{\mathrm{BWT} (s)}}
\newcommand{\tag}[2]{\ensuremath{\mathrm{#1}_{#2}}}

\begin{document}

\title{Bounds for Compression in Streaming Models
\thanks{Partially supported by the Italian MIUR and Italian-Israeli FIRB
project ``Pattern Discovery Algorithms in Discrete Structures, with
Applications to Bioinformatics''.}}
\author{}
\institute{}
\maketitle

\vspace{-10ex}

\begin{abstract}
Compression algorithms and streaming algorithms are both powerful tools for dealing with massive data sets, but many of the best compression algorithms --- e.g., those based on the Burrows-Wheeler Transform --- at first seem incompatible with streaming.  In this paper we consider several popular streaming models and ask in which, if any, we can compress as well as we can with the BWT.  We first prove a nearly tight tradeoff between memory and redundancy for the Standard, Multi\-pass and W-Streams models, demonstrating a bound that is achievable with the BWT but unachievable in those models.  We then show we can compute the related Schindler Transform in the StreamSort model and the BWT in the Read-Write model and, thus, achieve that bound.
\end{abstract}

\section{Introduction} \label{sec:intro}

The increasing size of data sets over the past decade has inspired work on both data compression and streaming algorithms.  In compression research, the advent of the Burrows-Wheeler Transform~\cite{BW94} (BWT) has led to great improvements in both theory and practice.  Streaming algorithms, meanwhile, are now used not only to process data online but also, because sequential access is so much faster than random access, to process them on disk more quickly.  To combine these advances, it seems we must find a way to compute something like the BWT in a streaming model, e.g.:
\begin{description}
\item[1. Standard:] In the simplest and most restrictive model, we are allowed only one pass over the input and memory sublinear (usually polylogarithmic) in the input's size (see, e.g.,~\cite{Mut05}).
\item[2. Multipass:] In one of the earliest papers on what is now called streaming, Munro and Paterson~\cite{MP80} proposed a model in which the input is stored on a one-way, read-only tape --- representing external memory --- that is completely rewound whenever we reach the end.
\item[3. W-Streams:] Ruhl~\cite{Ruh03} proposed a model in which we can also write: during each pass over the tape, we can replace its contents with something up to a constant factor larger.
\item[4. StreamSort:] Since we cannot sort in the W-Streams model with polylogarithmic memory and passes, Ruhl also proposed a generalization in which we can sort the contents of the tape at the cost of a constant number of passes (see also~\cite{ADRR04}).
\item[5. Read-Write:] Grohe and Schweikardt~\cite{GS05} noted that, with an additional tape, we can sort even in the W-Streams model; they proposed a model in which we have a read-write input tape, some number of read-write work tapes, and a write-only output tape.
\end{description}

In a previous paper~\cite{GM07b} we proved nearly tight bounds on how well we can compress in the Standard model with constant memory.  We also showed how LZ77~\cite{ZL77} can be implemented with a growing sliding window to use slightly sublinear memory without increasing the bound on its redundancy, but that this cannot be done for sublogarithmic memory.  Those bounds are, however, tangential to the common assumption of polylogarithmic memory.  In this paper we consider whether the following bound can be achieved in the models described above with polylogarithmic resources: Given a string $s$ of length $n$ over an alphabet of constant size $\sigma$, can we store $s$ in \(O (n H_k (s) + \sigma^k \log n)\) bits for all $k$ simultaneously, where \(H_k (s)\) is the $k$th-order empirical entropy of $s$?  We first prove the bound unachievable in the first three models, via a nearly tight tradeoff between memory and redundancy.  We then show we can compute the Schindler Transform~\cite{Sch97} (ST) in the StreamSort model and the BWT in the Read-Write model, so the bound is achievable in them.

We start by reviewing some preliminary material in Section~\ref{sec:prelims}.  In Section~\ref{sec:standard} we show how, for any constants $c$ and $\epsilon$ with \(1 \geq c \geq 0\) and \(\epsilon > 0\), we can store $s$ in \(n H_k (s) + O (\sigma^k n^{1 - c + \epsilon})\) bits in the Standard model with \(O (n^c)\) bits of memory; we then show this bound is nearly optimal in the sense that we cannot always store $s$ in, nor recover it from, \(O (n H_k (s) + \sigma^k n^{1 - c - \epsilon})\) bits.  In Section~\ref{sec:MP and WS} we extend our tradeoff to the Multipass and W-Streams models.  Since we can store \BWTs\ in, and recover it from, \(3.4 n H_k (s) + O (\sigma^k)\) bits in even the Standard model, our tradeoff implies we can neither compute nor invert the BWT in the first three models.  In Section~\ref{sec:StreamSort} we show how we can compute the ST in the StreamSort model and in Section~\ref{sec:ReadWrite} we show how we can compute the BWT in the Read-Write model.  Using either transform, we can store $s$ in \(1.8 n H_k (s) + O (\sigma^k \log n)\) bits in these models.

\section{Preliminaries} \label{sec:prelims}

Throughout this paper, we assume \(s = s_1 \cdots s_n\) is a string of length $n$ over an alphabet of constant size $\sigma$, and $c$ and $\epsilon$ are constants with \(1 - \epsilon > c > \epsilon > 0\).  The 0th-order empirical entropy \(H_0 (s)\) of $s$ is the entropy of the characters' distribution in $s$, i.e., \(H_0 (s) = \frac{1}{n} \sum_{a \in s} n_a \log \frac{n}{n_a}\) where \(a \in s\) means character $a$ occurs in $s$ and $n_a$ is its frequency.  The $k$th-order empirical entropy \(H_k (s)\) of $s$ for \(k \geq 1\), described in detail by Manzini~\cite{Man01}, is defined as
\[H_k (s) = \frac{1}{n} \sum_{|w| = k} |w_s| H_0 (w_s)\,,\]
where $w_s$ is the concatenation of characters that immediately follow occurrences of $w$ in $s$.

The BWT is an invertible transform that permutes the characters in $s$ so that $s_i$ comes before $s_j$ if \(s_{i - 1} \cdots s_{i - n}\) is lexicographically less than \(s_{j - 1} \cdots s_{j - n}\), taking indices modulo \(n + 1\) and considering \(s_0 = s_{n + 1}\) to be lexicographically less than any character in the alphabet.  In other words, the BWT sorts $s$'s characters in order of their contexts, which start at their predecessors and extend backwards.  Although the distribution of characters remains the same --- so \(H_0 (\BWTs) = H_0 (s)\) --- characters with similar contexts are moved close together so, if $s$ is compressible with an algorithm that takes uses contexts, then \BWTs\ is compressible with an algorithm that takes advantage of local homogeneity.  Move-to-front~\cite{BSTW86} and distance coding~\cite{Bin00,Deo02} are two reversible transforms that turn strings with local homogeneity into strings of numbers with low 0th-order empirical entropy: move-to-front keeps a list of the characters in the alphabet and replaces each character in the input by its position in the list and then moves it to the front of the list; distance coding writes the distance to the first occurrence of each character and the length of the string, then replaces each character in the input by the distance from the last occurrence of that character, omitting 1s.  (The numbers are often written with, e.g., Elias' delta code~\cite{Eli75}, in which case move-to-front and distance coding become compression algorithms themselves.)  Building on work by Kaplan, Landau and Verbin~\cite{KLV07}, we showed in another previous paper~\cite{GM07a} that composing the BWT, move-to-front, run-length coding and arithmetic coding produces an encoding that contains at most \(3.4 n H_k (s) + O (\sigma^k)\) bits for all $k$ simultaneously; with the BWT, distance coding and arithmetic coding, the bound is \(1.8 n H_k (s) + O (\sigma^k \log n)\) bits.  Using the BWT in more sophisticated ways, Ferragina, Manzini, M\"{a}kinen and Navarro~\cite{FMMN07} and M\"{a}kinen and Navarro~\cite{MN??} achieved a bound on the encoding's length of \(n H_k (s) + o (n / \log n)\) bits for all $k$ at most a constant proper fraction of \(\log_\sigma n\); Grossi, Gupta and Vitter~\cite{GGV??} achieved a bound of \(n H_k (s) + O (\sigma^k \log n)\) bits for all $k$ simultaneously, matching a lower bound due to Rissanen~\cite{Ris83}. In Section~\ref{sec:standard} we show how, because length times empirical entropy is superadditive --- i.e., \(|a| H_k (a) + |b| H_k (b) \leq |ab| H_k (ab)\) --- we can trade-off between memory and redundancy by breaking $s$ into blocks and encoding each block in turn with Grossi, Gupta and Vitter's algorithm: the memory needed decreases by a factor roughly equal to the number of blocks, while the redundancy increases by that much.  In Section~\ref{sec:StreamSort} we use the ST instead of the BWT.  When using contexts of length $k$, the ST permutes the characters of $s$ so that $s_i$ comes before $s_j$ if \(s_{i - 1} \cdots s_{i - k}\) is lexicographically less than \(s_{j - 1} \cdots s_{j - k}\) or, when they are equal, if \(i < j\).  By the same arguments as for the BWT, composing the ST, distance coding and arithmetic coding produces an encoding the contains at most \(1.8 n H_k (s) + O (\sigma^k \log n)\) bits for the chosen context length $k$.

A $\sigma$-ary De Bruijn sequence of order \(k \geq 1\) contains each possible $k$-tuple exactly once and, it follows, has length \(\sigma^k + k - 1\) and starts and ends with the same \(k - 1\) characters.  Notice that, if $d$ consists of the first $\sigma^k$ characters of such a sequence, then \(H_k (d^i) = 0\) for any $i$.  However, there are \((\sigma!)^{\sigma^{k - 1}}\) such sequences~\cite{DeB46,Knu97,Knu05} so, by Li and Vit\'{a}nyi's Incompressibility Lemma~\cite{LV97}, for any fixed algorithm $A$, the maximum Kolmogorov complexity of $d$ relative to $A$ is at least \(\lfloor \sigma^{k - 1} \log \sigma! \rfloor = \Omega (\sigma^k)\) bits.  We used this fact in~\cite{GM07b} to prove that a one-pass algorithm cannot always compress well in terms of the $k$th-order empirical without using memory exponential in $k$: we can compute $d$ from the configurations of $A$'s memory when it starts and finishes reading any copy of $d$ in $d^i$ and its output while reading that copy (if there were another string $d'$ that took $A$ between those configurations while producing that output, then we could substitute $d'$ for that copy of $d$ without changing the total encoding); therefore, if $A$ uses \(o (\sigma^k)\) bits of memory, then its total output must be \(\Omega (|d^i|)\) bits.

\section{The Standard model} \label{sec:standard}

We can easily store $s$ in \(n H_0 (s) + 2 n + O (1)\) bits in one pass with \(O (\log n)\) bits of memory with dynamic Huffman coding~\cite{Vit87}, for example; by running a separate copy of the algorithm for each possible $k$-tuple for any given $k$, we can store $s$ in \(n H_k (s) + 2 n + O (\sigma^k)\) bits in one pass with \(O (\sigma^k \log n)\) bits of memory.  With adaptive arithmetic coding instead of dynamic Huffman coding, the \(2 n\) term becomes approximately \(n / 100\) (see, e.g.,~\cite{HV92}).  With the version of LZ77 we described in~\cite{GM07b}, for any fixed $k$ we can store $s$ in \(n H_k (s) + o (n)\) bits in one pass with \(o (n)\) bits of memory, but with both \(o (n)\) terms nearly linear.  We now show we can substantially reduce those terms simultaneously, answering a question we posed in that paper.

\begin{theorem} \label{thm:Std Enc UB}
In the Standard model with \(O (n^c)\) bits of memory, we can store $s$ in, and later recover it from, \(n H_k (s) + O (\sigma^k n^{1 - c + \epsilon})\) bits for all $k$ simultaneously.
\end{theorem}

\begin{proof}
Let $A$ be Grossi, Gupta and Vitter's algorithm.  They proved $A$ stores $s$ in \(n H_k (s) + O (\sigma^k \log n)\) bits for all $k$ simultaneously using \(O (n)\) time so, although they did not give an explicit bound on the memory used, we know it is at most $A$'s time complexity multiplied by the word size, i.e., \(O (n \log n)\) bits.

First, suppose we know $n$ in advance.  We process $s$ in \(O (n^{1 - c + \epsilon / 2})\) blocks \(b_1, \ldots, b_m\), each of length \(O (n^{c - \epsilon / 2})\): we read each block $b_i$ in turn, compute and output \(A (b_i)\) --- using \(O (|b_i| \log |b_i|) = O (n^c)\) bits of memory --- and erase $b_i$ from memory.  As we noted in Section~\ref{sec:prelims}, empirical entropy is superadditive, so the total length of the encodings we output is at most
\[\sum_{i = 1}^m \left( \rule{0ex}{2ex} |b_i| H_k (b_i) + O (\sigma^k \log |b_i|) \right)
\leq n H_k (s) + O (\sigma^k n^{1 - c + \epsilon})\]
bits for all $k$ simultaneously.

Now, suppose we do not know $n$ in advance.  We work as before but we start with a constant estimate of $n$ and, each time we have read that many characters of $s$, we double it.  This way, we increase the size of the largest block by less than 2 and the number of blocks by an \(O (\log n)\)-factor, so our asymptotic bounds on the memory used and the whole encoding's length does not change. \qed
\end{proof}

Extending our lower bounds from that paper, we now show we cannot reduce the factor $n^{1 - c + \epsilon}$ much further unless we increase the factor $\sigma^k$, not even if we multiply the bound on the encoding's length by any constant coefficient.  This lower bound holds for both compression and decompression; as we noted in~\cite{GM07b}, ``good bounds [for decompression] are equally important, because often data is compressed once by a powerful machine (e.g., a server or base-station) and then transmitted to many weaker machines (clients or agents) who decompress it individually.''

\begin{theorem} \label{thm:Std Enc LB}
In the Standard model with \(O (n^c)\) bits of memory, we cannot store $s$ in \(O (n H_k (s) + \sigma^k n^{1 - c - \epsilon})\) bits in the worst case for, e.g., \(k = \lceil (c + \epsilon / 2)\) \(\log_\sigma n \rceil\).
\end{theorem}

\begin{proof}
Consider any compression algorithm $A$ that works in the Standard model with \(O (n^c)\) bits of memory.  Suppose $s$ consists of copies of the first $\sigma^k$ characters $d$ in a $\sigma$-ary De Bruijn sequence of order \(k = \lceil (c + \epsilon / 2) \log_\sigma n \rceil\) whose Kolmogorov complexity relative to $A$ is \(\Omega (\sigma^k) = \Omega (n^{c + \epsilon / 2})\) bits.  As we noted in Section~\ref{sec:prelims}, we can compute $d$ from the configurations of $A$'s memory when it starts and finishes reading any copy of $d$ and its output while reading that copy.  Since $A$'s memory is asymptotically smaller than $d$'s Kolmogorov complexity relative to $A$, it must output \(\Omega (\sigma^k)\) bits for every copy of $d$ in $s$, i.e., \(\Omega (n)\) bits altogether.  Since \(H_k (s) = 0\), however, \(O (n H_k (s) + \sigma^k n^{1 - c - \epsilon}) = O (n^{1 - \epsilon / 2})\). \qed
\end{proof}

\begin{theorem} \label{thm:Std Dec LB}
In the Standard model with \(O (n^c)\) bits of memory, we cannot recover $s$ from \(O (n H_k (s) + \sigma^k n^{1 - c - \epsilon})\) bits in the worst case for, e.g., \(k = \lceil (c + \epsilon / 2)\) \linebreak \(\log_\sigma n \rceil\).
\end{theorem}

\begin{proof}
Let $A$ be a decompression algorithm that works in the Standard model with \(O (n^c)\) bits of memory and let $s$ be as in the proof of Theorem~\ref{thm:Std Enc LB}.  We can also compute $d$ from the configuration of $A$'s memory when it starts outputting any copy of $d$ and the bits it reads while outputting that copy; since $A$'s memory is asymptotically smaller than $d$'s Kolmogorov complexity relative to $A$, it must read \(\Omega (\sigma^k)\) bits for each copy of $d$ in $s$, i.e., \(\Omega (n)\) bits altogether, whereas \(O (n H_k (s) + \sigma^k n^{1 - c - \epsilon}) = O (n^{1 - \epsilon / 2})\). \qed
\end{proof}

Finally, we now prove that, if we could compute the BWT in the Standard model, then we could achieve the bounds we have just proven unachievable; we will draw the obvious conclusion in Section~\ref{sec:MP and WS} as part of a more general theorem.

\begin{lemma} \label{lem:MTF RL UB}
In the Standard model with \(O (\log n)\) bits of memory, we can store \BWTs\ in, and later recover it from, \(3.4 n H_k (s) + O (\sigma^k)\) bits for all $k$ simultaneously.
\end{lemma}

\begin{proof}
We encode \BWTs\ by composing move-to-front, run-length coding and adaptive arithmetic coding; since encoding or decoding each of the three steps takes one pass and \(O (\log n)\) bits of memory, so does their composition.  As we noted in Section~\ref{sec:prelims}, the resulting encoding contains at most \(3.4 n H_k (s) + O (\sigma^k)\) bits for all $k$ simultaneously. \qed
\end{proof}

\section{The Multipass and W-Streams models} \label{sec:MP and WS}

We now extend our tradeoff to the Multipass and W-Streams models.  Of course, anything we can do in the Standard model we can do in those models, so the upper bounds are immediate.

\begin{theorem} \label{thm:MP and WS Enc UB}
In the Multipass and W-Streams models with \(O (n^c)\) bits of memory and one pass, we can store $s$ in \(n H_k (s) + O (\sigma^k n^{1 - c + \epsilon})\) bits for all $k$ simultaneously.
\end{theorem}

Conversely, lower bounds for the W-Streams model apply to the Multipass model.  We could quite easily extend our proofs to include the Multipass model alone: e.g., to see we cannot compress $s$ very well with polylogarithmic passes, notice we can compute $d$ from the configurations of $A$'s memory when it starts and finishes reading any copy of $d$ and its output while reading that copy \emph{during each pass}; since \(\log^{O (1)} n = o (n^{\epsilon / 2})\), even a polylogarithmic number of $A$'s memory configurations are asymptotically smaller than $d$'s Kolmogorov complexity relative to $A$, so the rest of our argument still holds.  The proof for the W-Streams model must be slightly different because, for example, after the first pass the tape will generally not contain copies of $d$.

\begin{theorem} \label{thm:MP and WS Enc LB}
In the Multipass and W-Streams models with \(O (n^c)\) bits of memory and \(\log^{O (1)} n\) passes, we cannot store $s$ in \(O (n H_k (s) + \sigma^k n^{1 - c - \epsilon})\) bits in the worst case for, e.g., \(k = \lceil (c + \epsilon / 2)\) \(\log_\sigma n \rceil\).
\end{theorem}

\begin{proof}
We need consider only the more general W-Streams model.  Consider any compression algorithm $A$ that works in the W-Streams model with \(O (n^c)\) bits of memory and \(\log^{O (1)} n\) passes, and let $s$ be as in the proof of Theorem~\ref{thm:Std Enc LB}.  Without loss of generality, assume $A$'s output consists of the contents of the tape after its last pass (any output from intermediate passes can be written on the tape instead).  Notice any substring of characters on the tape immediately after a particular pass must have been written (or left untouched) while $A$ was reading a substring of characters on the tape immediately before that pass.  Suppose $A$ makes $p$ passes over $s$.  Consider any copy $d_0$ of $d$ in $s$ and, for \(p \geq i \geq 1\), let $d_i$ be the substring $A$ writes while reading $d_{i - 1}$.  We claim we can compute $d$ from $d_p$ and the memory configurations of $A$ when it starts and finishes reading each $d_i$.  To see why, suppose there were a sequence \(d_0' \neq d_0, d_1', \ldots, d_{p - 1}', d_p' = d_p\) such that, for \(p \geq i \geq 1\), $d_{i - 1}'$ took $A$ between the $i$th pair of memory configurations while writing $d_i'$; then we could substitute $d_0'$ for $d_0$ without changing the total encoding.  It follows that $d_p$ must contain \(\Omega (\sigma^k)\) bits and, so, the whole tape must contain \(\Omega (n)\) bits after the last pass, whereas \(O (n H_k (s) + \sigma^k n^{1 - c - \epsilon}) = O (n^{1 - \epsilon / 2})\). \qed
\end{proof}

\begin{theorem} \label{thm:MP and WS Dec LB}
In the Multipass and W-Streams models with \(O (n^c)\) bits of memory and \(\log^{O (1)} n\) passes, we cannot recover $s$ from \(O (n H_k (s) + \sigma^k n^{1 - c - \epsilon})\) bits in the worst case for, e.g., \(k = \lceil (c + \epsilon / 2)\) \(\log_\sigma n \rceil\).
\end{theorem}

\begin{proof}
Again, we consider only the W-Streams model.  Consider any decompression algorithm $A$ that works in the W-Streams model with \(O (n^c)\) bits of memory and \(\log^{O (1)} n\) passes, and let $s$ be as in the proof of Theorem~\ref{thm:Std Enc LB}.  Without loss of generality, assume the tape contains $s$ after $A$'s last pass.  Consider any copy $d_0$ of $d$ in $s$ and, for \(p \geq i \geq 1\), let $d_i$ be the substring $A$ reads while writing $d_{i - 1}$; notice this is the reverse of the definition in the proof of Theorem~\ref{thm:MP and WS Enc LB}.  Notice we can compute $d$ from $d_p$ and the memory configuration of $A$ when it starts reading each $d_i$: running $A$ on $d_i$, starting in the memory configuration, produces $d_{i - 1}$.  It follows that $A$ must read \(\Omega (\sigma^k)\) bits for each copy of $d$ in $s$, i.e., \(\Omega (n)\) bits altogether, whereas \(O (n H_k (s) + \sigma^k n^{1 - c - \epsilon}) = O (n^{1 - \epsilon / 2})\). \qed
\end{proof}

We note in passing that in the W-Streams model, we can easily reduce sorting the characters in $s$ to computing the BWT: we compute \(s' = (s_1, s_0) (s_2, s_1) \cdots\) \((s_n, s_{n - 1}) (s_0, s_n)\); we compute \(\mathrm{BWT} (s') = (s_1, s_0) (s_{i_1 + 1}, s_{i_1}) \cdots (s_{i_n + 1}, s_{i_n})\); and we output \(s_{i_1}, \ldots, s_{i_n}\).  To see why \(s_{i_1}, \ldots, s_{i_n}\) are sorted, suppose \((s_{i + 1}, s_i)\) precedes \((s_{j + 1}, s_j)\) in \(\mathrm{BWT} (s')\).  The BWT arranges the pairs in $s'$ in the lexicographic order of their predecessors (how it breaks ties does not concern us now), which is that of their predecessors' first components, or that of their own second components --- so \(s_i \leq s_j\).  If $\sigma$ were unbounded, this reduction would imply we could not compute the BWT in the first three models; since $\sigma$ is constant, however, it is meaningless --- we can sort the characters in $s$ in the Multipass model anyway.  Fortunately, we can also easily reduce storing $s$ in \(O (n H_k (s) + \sigma^k \log n)\) bits to computing the BWT.

\begin{theorem} \label{cor:BWT LB}
In the Standard, Multipass and W-Streams models with \(\log^{O (1)} n\) bits of memory and passes, we can neither compute nor invert \BWTs\ in the worst case.
\end{theorem}

\begin{proof}
If we could compute or invert \BWTs\ then, by Lemma~\ref{lem:MTF RL UB}, we could store $s$ in, or recover it from, \(3.4 n H_k (s) + O (\sigma^k)\) bits for all $k$ simultaneously; however, by Theorems~\ref{thm:Std Enc LB},~\ref{thm:Std Dec LB},~\ref{thm:MP and WS Enc LB} and~\ref{thm:MP and WS Dec LB} we cannot achieve this bound in the worst case. \qed
\end{proof}

\section{The StreamSort model} \label{sec:StreamSort}

The ST is known as both the ``Schindler Transform'' and the ``Sort Transform'', so it is perhaps not surprising that it can be computed in the StreamSort model.  Indeed, computing the ST for any given \(k = O (\log n)\) takes only a constant number of passes once we have padded the input from \(O (n)\) bits to \(\Omega (n \log n)\) bits; we do this padding, which takes \(O (\log \log n)\) passes, because we are allowed to expand the tape contents by a only constant factor during each pass, and we want to eventually associate each character with a \(O (\log n)\)-bit key --- the $k$-tuple that precedes that character in $s$ --- and then stably sort by the keys.  Unfortunately, we do not know yet how to invert the ST in either the StreamSort or Read-Write models.

\begin{lemma} \label{lem:SS ST}
For any given \(k = O (\log n)\), we can compute \(\mathrm{ST} (s)\) in the StreamSort model with \(O (\log n)\) bits of memory and \(O (\log \log n)\) passes.
\end{lemma}

\begin{proof}
We make \(O (\log \log n)\) passes, each time doubling the length of each character's representation by padding it with 0s, until each character takes \(\Omega (\log n)\) bits.  We make another pass to associate each character in $s$ with the $k$-tuple that precedes it; since \(\log \sigma^k = O (\log n)\), we can use \(O (\log n)\) bits of memory to keep track of the last $k$ characters we have seen and write them as a key in front of the next character while only doubling the number of bits on the tape.  We then use \(O (1)\) passes to stably sort by those keys --- i.e., computing the ST --- and, finally, delete the keys and padding. \qed
\end{proof}

\begin{theorem} \label{thm:SS Enc UB}
In the StreamSort model with \(O (\log n)\) bits of memory and \linebreak \(O (\log n \log \log n)\) passes, we can store $s$ in \(1.8 n H_k (s) + O (\sigma^k \log n)\) bits for all $k$ simultaneously.
\end{theorem}

\begin{proof}
We compute \(\mathrm{ST} (s)\) in \(O (\log \log n)\) passes and encode it in \(O (1)\) passes by composing distance coding and adaptive arithmetic coding, for each \(k = O (\log n)\); in total, we use \(O (\log n)\) bits of memory and \(O (\log n \log \log n)\) passes.  As we noted in Section~\ref{sec:prelims}, each resulting encoding contains at most \(1.8 n H_k (s) + O (\sigma^k \log n)\) bits for the value of $k$ used to compute it; thus, the shortest encoding contains \(1.8 n H_k (s) + O (\sigma^k \log n)\) for all \(k = O (\log n)\) simultaneously and so --- because
\[1.8 n H_0 (s) + O (\log n) < 1.8 n H_k (s) + (\sigma^k \log n) = \omega (n)\]
for all \(k = \omega (\log n)\) --- for all $k$ simultaneously. \qed
\end{proof}

\section{The Read-Write model} \label{sec:ReadWrite}

Anything we can do in the StreamSort model we can do in the Read-Write model using an \(O (\log n)\)-factor more passes, so Theorem~\ref{thm:SS Enc UB} implies we can store $s$ in \(1.8 n H_k (s) + (\sigma^k \log n)\) bits for all $k$ simultaneously in the Read-Write model with \(O (\log n)\) bits of memory and \(O (\log^2 n \log \log n)\) passes.  It does not, however, imply we can recover $s$ again in the Read-Write model.  Fortunately, using techniques based on the doubling algorithm by Arge, Ferragina, Grossi and Vitter~\cite{AFGV97} for sorting strings in external memory (see also~\cite{DKMS??}), we can both compute and invert \BWTs\ in the Read-Write model.  Figures~\ref{fig:BWT Enc} and~\ref{fig:BWT Dec} show how we compute and invert \(\mathrm{BWT} (\mathrm{mississippi})\); to save space, Figure~\ref{fig:BWT Dec} shows two rounds of the algorithm in each row.

\begin{figure}
\begin{center}
\resizebox{.7\textwidth}{!}
{\begin{tabular}{c|c|c|c|c}
triples & copy and sort & merge & sort & shrink\\
\hline
\begin{tabular}{ccc}
&&\\
\tag{m}{1} & 1 & \tag{i}{2}\\
\tag{i}{2} & 1 & \tag{s}{3}\\
\tag{s}{3} & 1 & \tag{s}{4}\\
\tag{s}{4} & 1 & \tag{i}{5}\\
\tag{i}{5} & 1 & \tag{s}{6}\\
\tag{s}{6} & 1 & \tag{s}{7}\\
\tag{s}{7} & 1 & \tag{i}{8}\\
\tag{i}{8} & 1 & \tag{p}{9}\\
\tag{p}{9} & 1 & \tag{p}{10}\\
\tag{p}{10} & 1 & \tag{i}{11}\\
\tag{i}{11} & 1 & \tag{\#}{12}\\
\tag{\#}{12} & 1 & \tag{m}{1}\\
&&
\end{tabular} &
\begin{tabular}{ccccccc}
&&&&&&\\
\tag{i}{11} & 1 & \tag{\#}{12} & \rule{2ex}{0ex} & \tag{\#}{12} & 1 & \tag{m}{1}\\
\tag{m}{1} & 1 & \tag{i}{2} && \tag{i}{2} & 1 & \tag{s}{3}\\
\tag{s}{4} & 1 & \tag{i}{5} && \tag{i}{5} & 1 & \tag{s}{6}\\
\tag{s}{7} & 1 & \tag{i}{8} && \tag{i}{8} & 1 & \tag{p}{9}\\
\tag{p}{10} & 1 & \tag{i}{11} && \tag{i}{11} & 1 & \tag{\#}{12}\\
\tag{\#}{12} & 1 & \tag{m}{1} && \tag{m}{1} & 1 & \tag{i}{2}\\
\tag{i}{8} & 1 & \tag{p}{9} && \tag{p}{9} & 1 & \tag{p}{10}\\
\tag{p}{9} & 1 & \tag{p}{10} && \tag{p}{10} & 1 & \tag{i}{11}\\
\tag{i}{2} & 1 & \tag{s}{3} && \tag{s}{3} & 1 & \tag{s}{4}\\
\tag{s}{3} & 1 & \tag{s}{4} && \tag{s}{4} & 1 & \tag{i}{5}\\
\tag{i}{5} & 1 & \tag{s}{6} && \tag{s}{6} & 1 & \tag{s}{7}\\
\tag{s}{6} & 1 & \tag{s}{7} && \tag{s}{7} & 1 & \tag{i}{8}\\
&&&&&&
\end{tabular} &
\begin{tabular}{ccccc}
&&&&\\
\tag{i}{11} & 1 & \tag{\#}{12} & 1 & \tag{m}{1}\\
\tag{m}{1} & 1 & \tag{i}{2} & 1 & \tag{s}{3}\\
\tag{s}{4} & 1 & \tag{i}{5} & 1 & \tag{s}{6}\\
\tag{s}{7} & 1 & \tag{i}{8} & 1 & \tag{p}{9}\\
\tag{p}{10} & 1 & \tag{i}{11} & 1 & \tag{\#}{12}\\
\tag{\#}{12} & 1 & \tag{m}{1} & 1 & \tag{i}{2}\\
\tag{i}{8} & 1 & \tag{p}{9} & 1 & \tag{p}{10}\\
\tag{p}{9} & 1 & \tag{p}{10} & 1 & \tag{i}{11}\\
\tag{i}{2} & 1 & \tag{s}{3} & 1 & \tag{s}{4}\\
\tag{s}{3} & 1 & \tag{s}{4} & 1 & \tag{i}{5}\\
\tag{i}{5} & 1 & \tag{s}{6} & 1 & \tag{s}{7}\\
\tag{s}{6} & 1 & \tag{s}{7} & 1 & \tag{i}{8}\\
&&&&
\end{tabular} &
\begin{tabular}{ccccc}
&&&&\\
\tag{i}{11} & 1 & \tag{\#}{12} & 1 & \tag{m}{1}\\
\tag{m}{1} & 1 & \tag{i}{2} & 1 & \tag{s}{3}\\
\tag{s}{4} & 1 & \tag{i}{5} & 1 & \tag{s}{6}\\
\tag{s}{7} & 1 & \tag{i}{8} & 1 & \tag{p}{9}\\
\tag{p}{10} & 1 & \tag{i}{11} & 1 & \tag{\#}{12}\\
\tag{\#}{12} & 1 & \tag{m}{1} & 1 & \tag{i}{2}\\
\tag{i}{8} & 1 & \tag{p}{9} & 1 & \tag{p}{10}\\
\tag{p}{9} & 1 & \tag{p}{10} & 1 & \tag{i}{11}\\
\tag{i}{2} & 1 & \tag{s}{3} & 1 & \tag{s}{4}\\
\tag{s}{3} & 1 & \tag{s}{4} & 1 & \tag{i}{5}\\
\tag{i}{5} & 1 & \tag{s}{6} & 1 & \tag{s}{7}\\
\tag{s}{6} & 1 & \tag{s}{7} & 1 & \tag{i}{8}\\
&&&&
\end{tabular} &
\begin{tabular}{ccccc}
&&&&\\
\tag{i}{11} & 1 & \tag{m}{1}\\
\tag{m}{1} & 2 & \tag{s}{3}\\
\tag{s}{4} & 2 & \tag{s}{6}\\
\tag{s}{7} & 2 & \tag{p}{9}\\
\tag{p}{10} & 2 & \tag{\#}{12}\\
\tag{\#}{12} & 3 & \tag{i}{2}\\
\tag{i}{8} & 4 & \tag{p}{10}\\
\tag{p}{9} & 4 & \tag{i}{11}\\
\tag{i}{2} & 5 & \tag{s}{4}\\
\tag{s}{3} & 5 & \tag{i}{5}\\
\tag{i}{5} & 5 & \tag{s}{7}\\
\tag{s}{6} & 5 & \tag{i}{8}\\
&&&&
\end{tabular}\\
\hline
\begin{tabular}{ccc}
&&\\
\tag{i}{11} & 1 & \tag{m}{1}\\
\tag{m}{1} & 2 & \tag{s}{3}\\
\tag{s}{4} & 2 & \tag{s}{6}\\
\tag{s}{7} & 2 & \tag{p}{9}\\
\tag{p}{10} & 2 & \tag{\#}{12}\\
\tag{\#}{12} & 3 & \tag{i}{2}\\
\tag{i}{8} & 4 & \tag{p}{10}\\
\tag{p}{9} & 4 & \tag{i}{11}\\
\tag{i}{2} & 5 & \tag{s}{4}\\
\tag{s}{3} & 5 & \tag{i}{5}\\
\tag{i}{5} & 5 & \tag{s}{7}\\
\tag{s}{6} & 5 & \tag{i}{8}\\
&&
\end{tabular} &
\begin{tabular}{ccccccc}
&&&&&&\\
\tag{p}{10} & 2 & \tag{\#}{12} & \rule{2ex}{0ex} & \tag{\#}{12} & 3 & \tag{i}{2}\\
\tag{\#}{12} & 3 & \tag{i}{2} && \tag{i}{2} & 5 & \tag{s}{4}\\
\tag{s}{3} & 5 & \tag{i}{5} && \tag{i}{5} & 5 & \tag{s}{7}\\
\tag{s}{6} & 5 & \tag{i}{8} && \tag{i}{8} & 4 & \tag{p}{10}\\
\tag{p}{9} & 4 & \tag{i}{11} && \tag{i}{11} & 1 & \tag{m}{1}\\
\tag{i}{11} & 1 & \tag{m}{1} && \tag{m}{1} & 2 & \tag{s}{3}\\
\tag{s}{7} & 2 & \tag{p}{9} && \tag{p}{9} & 4 & \tag{i}{11}\\
\tag{i}{8} & 4 & \tag{p}{10} && \tag{p}{10} & 2 & \tag{\#}{12}\\
\tag{m}{1} & 2 & \tag{s}{3} && \tag{s}{3} & 5 & \tag{i}{5}\\
\tag{i}{2} & 5 & \tag{s}{4} && \tag{s}{4} & 2 & \tag{s}{6}\\
\tag{s}{4} & 2 & \tag{s}{6} && \tag{s}{6} & 5 & \tag{i}{8}\\
\tag{i}{5} & 5 & \tag{s}{7} && \tag{s}{7} & 2 & \tag{p}{9}\\
&&&&&&
\end{tabular} &
\begin{tabular}{ccccc}
&&&&\\
\tag{p}{10} & 2 & \tag{\#}{12} & 3 & \tag{i}{2}\\
\tag{\#}{12} & 3 & \tag{i}{2} & 5 & \tag{s}{4}\\
\tag{s}{3} & 5 & \tag{i}{5} & 5 & \tag{s}{7}\\
\tag{s}{6} & 5 & \tag{i}{8} & 4 & \tag{p}{10}\\
\tag{p}{9} & 4 & \tag{i}{11} & 1 & \tag{m}{1}\\
\tag{i}{11} & 1 & \tag{m}{1} & 2 & \tag{s}{3}\\
\tag{s}{7} & 2 & \tag{p}{9} & 4 & \tag{i}{11}\\
\tag{i}{8} & 4 & \tag{p}{10} & 2 & \tag{\#}{12}\\
\tag{m}{1} & 2 & \tag{s}{3} & 5 & \tag{i}{5}\\
\tag{i}{2} & 5 & \tag{s}{4} & 2 & \tag{s}{6}\\
\tag{s}{4} & 2 & \tag{s}{6} & 5 & \tag{i}{8}\\
\tag{i}{5} & 5 & \tag{s}{7} & 2 & \tag{p}{9}\\
&&&&
\end{tabular} &
\begin{tabular}{ccccc}
&&&&\\
\tag{p}{9} & 4 & \tag{i}{11} & 1 & \tag{m}{1}\\
\tag{i}{11} & 1 & \tag{m}{1} & 2 & \tag{s}{3}\\
\tag{i}{8} & 4 & \tag{p}{10} & 2 & \tag{\#}{12}\\
\tag{i}{2} & 5 & \tag{s}{4} & 2 & \tag{s}{6}\\
\tag{i}{5} & 5 & \tag{s}{7} & 2 & \tag{p}{9}\\
\tag{p}{10} & 2 & \tag{\#}{12} & 3 & \tag{i}{2}\\
\tag{s}{6} & 5 & \tag{i}{8} & 4 & \tag{p}{10}\\
\tag{s}{7} & 2 & \tag{p}{9} & 4 & \tag{i}{11}\\
\tag{\#}{12} & 3 & \tag{i}{2} & 5 & \tag{s}{4}\\
\tag{s}{3} & 5 & \tag{i}{5} & 5 & \tag{s}{7}\\
\tag{m}{1} & 2 & \tag{s}{3} & 5 & \tag{i}{5}\\
\tag{s}{4} & 2 & \tag{s}{6} & 5 & \tag{i}{8}\\
&&&&
\end{tabular} &
\begin{tabular}{ccc}
&&\\
\tag{p}{9} & 1 & \tag{m}{1}\\
\tag{i}{11} & 2 & \tag{s}{3}\\
\tag{i}{8} & 3 & \tag{\#}{12}\\
\tag{i}{2} & 4 & \tag{s}{6}\\
\tag{i}{5} & 4 & \tag{p}{9}\\
\tag{p}{10} & 5 & \tag{i}{2}\\
\tag{s}{6} & 6 & \tag{p}{10}\\
\tag{s}{7} & 7 & \tag{i}{11}\\
\tag{\#}{12} & 8 & \tag{s}{4}\\
\tag{s}{3} & 9 & \tag{s}{7}\\
\tag{m}{1} & 10 & \tag{i}{5}\\
\tag{s}{4} & 10 & \tag{i}{8}\\
&&
\end{tabular}\\
\hline
\begin{tabular}{ccc}
&&\\
\tag{p}{9} & 1 & \tag{m}{1}\\
\tag{i}{11} & 2 & \tag{s}{3}\\
\tag{i}{8} & 3 & \tag{\#}{12}\\
\tag{i}{2} & 4 & \tag{s}{6}\\
\tag{i}{5} & 4 & \tag{p}{9}\\
\tag{p}{10} & 5 & \tag{i}{2}\\
\tag{s}{6} & 6 & \tag{p}{10}\\
\tag{s}{7} & 7 & \tag{i}{11}\\
\tag{\#}{12} & 8 & \tag{s}{4}\\
\tag{s}{3} & 9 & \tag{s}{7}\\
\tag{m}{1} & 10 & \tag{i}{5}\\
\tag{s}{4} & 10 & \tag{i}{8}
\end{tabular} &
\begin{tabular}{ccccccc}
&&&&&&\\
\tag{i}{8} & 3 & \tag{\#}{12} & \rule{2ex}{0ex} & \tag{\#}{12} & 8 & \tag{s}{4}\\
\tag{p}{10} & 5 & \tag{i}{2} && \tag{i}{2} & 4 & \tag{s}{6}\\
\tag{m}{1} & 10 & \tag{i}{5} && \tag{i}{5} & 4 & \tag{p}{9}\\
\tag{s}{4} & 10 & \tag{i}{8} && \tag{i}{8} & 3 & \tag{\#}{12}\\
\tag{s}{7} & 7 & \tag{i}{11} && \tag{i}{11} & 2 & \tag{s}{3}\\
\tag{p}{9} & 1 & \tag{m}{1} && \tag{m}{1} & 10 & \tag{i}{5}\\
\tag{i}{5} & 4 & \tag{p}{9} && \tag{p}{9} & 1 & \tag{m}{1}\\
\tag{s}{6} & 6 & \tag{p}{10} && \tag{p}{10} & 5 & \tag{i}{2}\\
\tag{i}{11} & 2 & \tag{s}{3} && \tag{s}{3} & 9 & \tag{s}{7}\\
\tag{\#}{12} & 8 & \tag{s}{4} && \tag{s}{4} & 10 & \tag{i}{8}\\
\tag{i}{2} & 4 & \tag{s}{6} && \tag{s}{6} & 6 & \tag{p}{10}\\
\tag{s}{3} & 9 & \tag{s}{7} && \tag{s}{7} & 7 & \tag{i}{11}
\end{tabular} &
\begin{tabular}{ccccc}
&&&&\\
\tag{i}{8} & 3 & \tag{\#}{12} & 8 & \tag{s}{4}\\
\tag{p}{10} & 5 & \tag{i}{2} & 4 & \tag{s}{6}\\
\tag{m}{1} & 10 & \tag{i}{5} & 4 & \tag{p}{9}\\
\tag{s}{4} & 10 & \tag{i}{8} & 3 & \tag{\#}{12}\\
\tag{s}{7} & 7 & \tag{i}{11} & 2 & \tag{s}{3}\\
\tag{p}{9} & 1 & \tag{m}{1} & 10 & \tag{i}{5}\\
\tag{i}{5} & 4 & \tag{p}{9} & 1 & \tag{m}{1}\\
\tag{s}{6} & 6 & \tag{p}{10} & 5 & \tag{i}{2}\\
\tag{i}{11} & 2 & \tag{s}{3} & 9 & \tag{s}{7}\\
\tag{\#}{12} & 8 & \tag{s}{4} & 10 & \tag{i}{8}\\
\tag{i}{2} & 4 & \tag{s}{6} & 6 & \tag{p}{10}\\
\tag{s}{3} & 9 & \tag{s}{7} & 7 & \tag{i}{11}
\end{tabular} &
\begin{tabular}{ccccc}
&&&&\\
\tag{i}{5} & 4 & \tag{p}{9} & 1 & \tag{m}{1}\\
\tag{s}{7} & 7 & \tag{i}{11} & 2 & \tag{s}{3}\\
\tag{s}{4} & 10 & \tag{i}{8} & 3 & \tag{\#}{12}\\
\tag{p}{10} & 5 & \tag{i}{2} & 4 & \tag{s}{6}\\
\tag{m}{1} & 10 & \tag{i}{5} & 4 & \tag{p}{9}\\
\tag{s}{6} & 6 & \tag{p}{10} & 5 & \tag{i}{2}\\
\tag{i}{2} & 4 & \tag{s}{6} & 6 & \tag{p}{10}\\
\tag{s}{3} & 9 & \tag{s}{7} & 7 & \tag{i}{11}\\
\tag{i}{8} & 3 & \tag{\#}{12} & 8 & \tag{s}{4}\\
\tag{i}{11} & 2 & \tag{s}{3} & 9 & \tag{s}{7}\\
\tag{p}{9} & 1 & \tag{m}{1} & 10 & \tag{i}{5}\\
\tag{\#}{12} & 8 & \tag{s}{4} & 10 & \tag{i}{8}
\end{tabular} &
\begin{tabular}{ccc}
&&\\
\tag{i}{5} & 1 & \tag{m}{1}\\
\tag{s}{7} & 2 & \tag{s}{3}\\
\tag{s}{4} & 3 & \tag{\#}{12}\\
\tag{p}{10} & 4 & \tag{s}{6}\\
\tag{m}{1} & 5 & \tag{p}{9}\\
\tag{s}{6} & 6 & \tag{i}{2}\\
\tag{i}{2} & 7 & \tag{p}{10}\\
\tag{s}{3} & 8 & \tag{i}{11}\\
\tag{i}{8} & 9 & \tag{s}{4}\\
\tag{i}{11} & 10 & \tag{s}{7}\\
\tag{p}{9} & 11 & \tag{i}{5}\\
\tag{\#}{12} & 12 & \tag{i}{8}
\end{tabular}
\end{tabular}}
\end{center}
\vspace{-3ex}
\caption{Computing the BWT in the Read-Write model.}
\label{fig:BWT Enc}
\vspace{2ex}
\begin{center}
\resizebox{.9\textwidth}{!}
{\begin{tabular}{c|c|c|c|c|c|c}
triples & copy and sort & merge & shrink & copy and sort & merge & shrink\\
\hline
\begin{tabular}{ccc}
&&\\
\tag{\#}{3} & 12 & \tag{m}{1}\\
\tag{i}{6} & ? & \tag{s}{2}\\
\tag{i}{8} & ? & \tag{\#}{3}\\
\tag{i}{11} & ? & \tag{s}{4}\\
\tag{i}{12} & ? & \tag{p}{5}\\
\tag{m}{1} & ? & \tag{i}{6}\\
\tag{p}{5} & ? & \tag{p}{7}\\
\tag{p}{7} & ? & \tag{i}{8}\\
\tag{s}{2} & ? & \tag{s}{9}\\
\tag{s}{4} & ? & \tag{s}{10}\\
\tag{s}{9} & ? & \tag{i}{11}\\
\tag{s}{10} & ? & \tag{i}{12}\\
&&
\end{tabular} &
\begin{tabular}{ccccccc}
&&&&&&\\
\tag{i}{8} & ? & \tag{\#}{3} & \rule{2ex}{0ex}& \tag{\#}{3} & 12 & \tag{m}{1}\\
\tag{m}{1} & ? & \tag{i}{6} && \tag{i}{6} & ? & \tag{s}{2}\\
\tag{p}{7} & ? & \tag{i}{8} && \tag{i}{8} & ? & \tag{\#}{3}\\
\tag{s}{9} & ? & \tag{i}{11} && \tag{i}{11} & ? & \tag{s}{4}\\
\tag{s}{10} & ? & \tag{i}{12} && \tag{i}{12} & ? & \tag{p}{5}\\
\tag{\#}{3} & 12 & \tag{m}{1} && \tag{m}{1} & ? & \tag{i}{6}\\
\tag{i}{12} & ? & \tag{p}{5} && \tag{p}{5} & ? & \tag{p}{7}\\
\tag{p}{5} & ? & \tag{p}{7} && \tag{p}{7} & ? & \tag{i}{8}\\
\tag{i}{6} & ? & \tag{s}{2} && \tag{s}{2} & ? & \tag{s}{9}\\
\tag{i}{11} & ? & \tag{s}{4} && \tag{s}{4} & ? & \tag{s}{10}\\
\tag{s}{2} & ? & \tag{s}{9} && \tag{s}{9} & ? & \tag{i}{11}\\
\tag{s}{4} & ? & \tag{s}{10} && \tag{s}{10} & ? & \tag{i}{12}\\
&&&&&&
\end{tabular} &
\begin{tabular}{ccccc}
&&&&\\
\tag{i}{8} & ? & \tag{\#}{3} & 12 & \tag{m}{1}\\
\tag{m}{1} & ? & \tag{i}{6} & ? & \tag{s}{2}\\
\tag{p}{7} & ? & \tag{i}{8} & ? & \tag{\#}{3}\\
\tag{s}{9} & ? & \tag{i}{11} & ? & \tag{s}{4}\\
\tag{s}{10} & ? & \tag{i}{12} & ? & \tag{p}{5}\\
\tag{\#}{3} & 12 & \tag{m}{1} & ? & \tag{i}{6}\\
\tag{i}{12} & ? & \tag{p}{5} & ? & \tag{p}{7}\\
\tag{p}{5} & ? & \tag{p}{7} & ? & \tag{i}{8}\\
\tag{i}{6} & ? & \tag{s}{2} & ? & \tag{s}{9}\\
\tag{i}{11} & ? & \tag{s}{4} & ? & \tag{s}{10}\\
\tag{s}{2} & ? & \tag{s}{9} & ? & \tag{i}{11}\\
\tag{s}{4} & ? & \tag{s}{10} & ? & \tag{i}{12}\\
&&&&
\end{tabular} &
\begin{tabular}{ccc}
&&\\
\tag{i}{8} & 11 & \tag{m}{1}\\
\tag{m}{1} & ? & \tag{s}{2}\\
\tag{p}{7} & ? & \tag{\#}{3}\\
\tag{s}{9} & ? & \tag{s}{4}\\
\tag{s}{10} & ? & \tag{p}{5}\\
\tag{\#}{3} & 12 & \tag{i}{6}\\
\tag{i}{12} & ? & \tag{p}{7}\\
\tag{p}{5} & ? & \tag{i}{8}\\
\tag{i}{6} & ? & \tag{s}{9}\\
\tag{i}{11} & ? & \tag{s}{10}\\
\tag{s}{2} & ? & \tag{i}{11}\\
\tag{s}{4} & ? & \tag{i}{12}\\
&&
\end{tabular} &
\begin{tabular}{ccccccc}
&&&&&&\\
\tag{p}{7} & ? & \tag{\#}{3} & \rule{2ex}{0ex} & \tag{\#}{3} & 12 & \tag{i}{6}\\
\tag{\#}{3} & 12 & \tag{i}{6} && \tag{i}{6} & ? & \tag{s}{9}\\
\tag{p}{5} & ? & \tag{i}{8} && \tag{i}{8} & 11 & \tag{m}{1}\\
\tag{s}{2} & ? & \tag{i}{11} && \tag{i}{11} & ? & \tag{s}{10}\\
\tag{s}{4} & ? & \tag{i}{12} && \tag{i}{12} & ? & \tag{p}{7}\\
\tag{i}{8} & 11 & \tag{m}{1} && \tag{m}{1} & ? & \tag{s}{2}\\
\tag{s}{10} & ? & \tag{p}{5} && \tag{p}{5} & ? & \tag{i}{8}\\
\tag{i}{12} & ? & \tag{p}{7} && \tag{p}{7} & ? & \tag{\#}{3}\\
\tag{m}{1} & ? & \tag{s}{2} && \tag{s}{2} & ? & \tag{i}{11}\\
\tag{s}{9} & ? & \tag{s}{4} && \tag{s}{4} & ? & \tag{i}{12}\\
\tag{i}{6} & ? & \tag{s}{9} && \tag{s}{9} & ? & \tag{s}{4}\\
\tag{i}{11} & ? & \tag{s}{10} && \tag{s}{10} & ? & \tag{p}{5}\\
&&&&&&
\end{tabular} &
\begin{tabular}{ccccc}
&&&&\\
\tag{p}{7} & ? & \tag{\#}{3} & 12 & \tag{i}{6}\\
\tag{\#}{3} & 12 & \tag{i}{6} & ? & \tag{s}{9}\\
\tag{p}{5} & ? & \tag{i}{8} & 11 & \tag{m}{1}\\
\tag{s}{2} & ? & \tag{i}{11} & ? & \tag{s}{10}\\
\tag{s}{4} & ? & \tag{i}{12} & ? & \tag{p}{7}\\
\tag{i}{8} & 11 & \tag{m}{1} & ? & \tag{s}{2}\\
\tag{s}{10} & ? & \tag{p}{5} & ? & \tag{i}{8}\\
\tag{i}{12} & ? & \tag{p}{7} & ? & \tag{\#}{3}\\
\tag{m}{1} & ? & \tag{s}{2} & ? & \tag{i}{11}\\
\tag{s}{9} & ? & \tag{s}{4} & ? & \tag{i}{12}\\
\tag{i}{6} & ? & \tag{s}{9} & ? & \tag{s}{4}\\
\tag{i}{11} & ? & \tag{s}{10} & ? & \tag{p}{5}\\
&&&&
\end{tabular} &
\begin{tabular}{ccc}
&&\\
\tag{p}{7} & 10 & \tag{i}{6}\\
\tag{\#}{3} & 12 & \tag{s}{9}\\
\tag{p}{5} & 9 & \tag{m}{1}\\
\tag{s}{2} & ? & \tag{s}{10}\\
\tag{s}{4} & ? & \tag{p}{7}\\
\tag{i}{8} & 11 & \tag{s}{2}\\
\tag{s}{10} & ? & \tag{i}{8}\\
\tag{i}{12} & ? & \tag{\#}{3}\\
\tag{m}{1} & ? & \tag{i}{11}\\
\tag{s}{9} & ? & \tag{i}{12}\\
\tag{i}{6} & ? & \tag{s}{4}\\
\tag{i}{11} & ? & \tag{p}{5}\\
&&
\end{tabular}\\
\hline
\begin{tabular}{ccc}
&&\\
\tag{p}{7} & 10 & \tag{i}{6}\\
\tag{\#}{3} & 12 & \tag{s}{9}\\
\tag{p}{5} & 9 & \tag{m}{1}\\
\tag{s}{2} & ? & \tag{s}{10}\\
\tag{s}{4} & ? & \tag{p}{7}\\
\tag{i}{8} & 11 & \tag{s}{2}\\
\tag{s}{10} & ? & \tag{i}{8}\\
\tag{i}{12} & ? & \tag{\#}{3}\\
\tag{m}{1} & ? & \tag{i}{11}\\
\tag{s}{9} & ? & \tag{i}{12}\\
\tag{i}{6} & ? & \tag{s}{4}\\
\tag{i}{11} & ? & \tag{p}{5}
\end{tabular} &
\begin{tabular}{ccccccc}
&&&&&&\\
\tag{i}{12} & ? & \tag{\#}{3} & \rule{2ex}{0ex} & \tag{\#}{3} & 12 & \tag{s}{9}\\
\tag{p}{7} & 10 & \tag{i}{6} && \tag{i}{6} & ? & \tag{s}{4}\\
\tag{s}{10} & ? & \tag{i}{8} && \tag{i}{8} & 11 & \tag{s}{2}\\
\tag{m}{1} & ? & \tag{i}{11} && \tag{i}{11} & ? & \tag{p}{5}\\
\tag{s}{9} & ? & \tag{i}{12} && \tag{i}{12} & ? & \tag{\#}{3}\\
\tag{p}{5} & 9 & \tag{m}{1} && \tag{m}{1} & ? & \tag{i}{11}\\
\tag{i}{11} & ? & \tag{p}{5} && \tag{p}{5} & 9 & \tag{m}{1}\\
\tag{s}{4} & ? & \tag{p}{7} && \tag{p}{7} & 10 & \tag{i}{6}\\
\tag{i}{8} & 11 & \tag{s}{2} && \tag{s}{2} & ? & \tag{s}{10}\\
\tag{i}{6} & ? & \tag{s}{4} && \tag{s}{4} & ? & \tag{p}{7}\\
\tag{\#}{3} & 12 & \tag{s}{9} && \tag{s}{9} & ? & \tag{i}{12}\\
\tag{s}{2} & ? & \tag{s}{10} && \tag{s}{10} & ? & \tag{i}{8}
\end{tabular} &
\begin{tabular}{ccccc}
&&&&\\
\tag{i}{12} & ? & \tag{\#}{3} & 12 & \tag{s}{9}\\
\tag{p}{7} & 10 & \tag{i}{6} & ? & \tag{s}{4}\\
\tag{s}{10} & ? & \tag{i}{8} & 11 & \tag{s}{2}\\
\tag{m}{1} & ? & \tag{i}{11} & ? & \tag{p}{5}\\
\tag{s}{9} & ? & \tag{i}{12} & ? & \tag{\#}{3}\\
\tag{p}{5} & 9 & \tag{m}{1} & ? & \tag{i}{11}\\
\tag{i}{11} & ? & \tag{p}{5} & 9 & \tag{m}{1}\\
\tag{s}{4} & ? & \tag{p}{7} & 10 & \tag{i}{6}\\
\tag{i}{8} & 11 & \tag{s}{2} & ? & \tag{s}{10}\\
\tag{i}{6} & ? & \tag{s}{4} & ? & \tag{p}{7}\\
\tag{\#}{3} & 12 & \tag{s}{9} & ? & \tag{i}{12}\\
\tag{s}{2} & ? & \tag{s}{10} & ? & \tag{i}{8}
\end{tabular} &
\begin{tabular}{ccc}
&&\\
\tag{i}{12} & 8 & \tag{s}{9}\\
\tag{p}{7} & 10 & \tag{s}{4}\\
\tag{s}{10} & 7 & \tag{s}{2}\\
\tag{m}{1} & ? & \tag{p}{5}\\
\tag{s}{9} & ? & \tag{\#}{3}\\
\tag{p}{5} & 9 & \tag{i}{11}\\
\tag{i}{11} & 5 & \tag{m}{1}\\
\tag{s}{4} & 6 & \tag{i}{6}\\
\tag{i}{8} & 11 & \tag{s}{10}\\
\tag{i}{6} & ? & \tag{p}{7}\\
\tag{\#}{3} & 12 & \tag{i}{12}\\
\tag{s}{2} & ? & \tag{i}{8}
\end{tabular} &
\begin{tabular}{ccccccc}
&&&&&&\\
\tag{s}{9} & ? & \tag{\#}{3} & \rule{2ex}{0ex} & \tag{\#}{3} & 12 & \tag{i}{12}\\
\tag{s}{4} & 6 & \tag{i}{6} && \tag{i}{6} & ? & \tag{p}{7}\\
\tag{s}{2} & ? & \tag{i}{8} && \tag{i}{8} & 11 & \tag{s}{10}\\
\tag{p}{5} & 9 & \tag{i}{11} && \tag{i}{11} & 5 & \tag{m}{1}\\
\tag{\#}{3} & 12 & \tag{i}{12} && \tag{i}{12} & 8 & \tag{s}{9}\\
\tag{i}{11} & 5 & \tag{m}{1} && \tag{m}{1} & ? & \tag{p}{5}\\
\tag{m}{1} & ? & \tag{p}{5} && \tag{p}{5} & 9 & \tag{i}{11}\\
\tag{i}{6} & ? & \tag{p}{7} && \tag{p}{7} & 10 & \tag{s}{4}\\
\tag{s}{10} & 7 & \tag{s}{2} && \tag{s}{2} & ? & \tag{i}{8}\\
\tag{p}{7} & 10 & \tag{s}{4} && \tag{s}{4} & 6 & \tag{i}{6}\\
\tag{i}{12} & 8 & \tag{s}{9} && \tag{s}{9} & ? & \tag{\#}{3}\\
\tag{i}{8} & 11 & \tag{s}{10} && \tag{s}{10} & 7 & \tag{s}{2}
\end{tabular} &
\begin{tabular}{ccccc}
&&&&\\
\tag{s}{9} & ? & \tag{\#}{3} & 12 & \tag{i}{12}\\
\tag{s}{4} & 6 & \tag{i}{6} & ? & \tag{p}{7}\\
\tag{s}{2} & ? & \tag{i}{8} & 11 & \tag{s}{10}\\
\tag{p}{5} & 9 & \tag{i}{11} & 5 & \tag{m}{1}\\
\tag{\#}{3} & 12 & \tag{i}{12} & 8 & \tag{s}{9}\\
\tag{i}{11} & 5 & \tag{m}{1} & ? & \tag{p}{5}\\
\tag{m}{1} & ? & \tag{p}{5} & 9 & \tag{i}{11}\\
\tag{i}{6} & ? & \tag{p}{7} & 10 & \tag{s}{4}\\
\tag{s}{10} & 7 & \tag{s}{2} & ? & \tag{i}{8}\\
\tag{p}{7} & 10 & \tag{s}{4} & 6 & \tag{i}{6}\\
\tag{i}{12} & 8 & \tag{s}{9} & ? & \tag{\#}{3}\\
\tag{i}{8} & 11 & \tag{s}{10} & 7 & \tag{s}{2}
\end{tabular} &
\begin{tabular}{ccc}
&&\\
\tag{s}{9} & 4 & \tag{i}{12}\\
\tag{s}{4} & 6 & \tag{p}{7}\\
\tag{s}{2} & 3 & \tag{s}{10}\\
\tag{p}{5} & 9 & \tag{m}{1}\\
\tag{\#}{3} & 12 & \tag{s}{9}\\
\tag{i}{11} & 5 & \tag{p}{5}\\
\tag{m}{1} & 1 & \tag{i}{11}\\
\tag{i}{6} & 2 & \tag{s}{4}\\
\tag{s}{10} & 7 & \tag{i}{8}\\
\tag{p}{7} & 10 & \tag{i}{6}\\
\tag{i}{12} & 8 & \tag{\#}{3}\\
\tag{i}{8} & 11 & \tag{s}{2}
\end{tabular}
\end{tabular}}
\end{center}
\vspace{-3ex}
\caption{Inverting the BWT in the Read-Write model.}
\label{fig:BWT Dec}
\end{figure}

\begin{theorem} \label{thm:RW UB}
In the Read-Write model with \(O (\log n)\) bits of memory and \(O (\log^2 n)\) passes, we can store $s$ in, and later recover it from, \(1.8 n H_k (s) + O (\sigma^k \log n)\) bits for all $k$ simultaneously.
\end{theorem}

\begin{proof}
As we noted in Section~\ref{sec:prelims}, once we have \BWTs, we can store it in \(1.8 n H_k (s) + O (\sigma^k \log n)\) bits for all $k$ simultaneously by composing distance coding and adaptive arithmetic coding.  To compute or invert distance coding and to encode or decode adaptive arithmetic coding all take \(O (\log n)\) bits of memory and \(O (1)\) passes.  Therefore, we need consider only how to compute and invert \BWTs.  Due to space constraints, here we only sketch these procedures; we will give full descriptions and analyses in the full paper.

To compute \BWTs, we append a special character that is lexicographically less than any character in the alphabet ($\#$ in Figures~\ref{fig:BWT Enc} and~\ref{fig:BWT Dec}).  We tag each character in $s$ with a unique identifier (e.g., its position; in this model, we can expand the contents of the tape by more than a constant factor during a pass, so we do not need to pad --- although \(O (\log \log n)\) extra rounds would not make a difference here anyway), then form a triple from it by appending a 1 and its successor in $s$.  We make two copies of the set of triples, sort the first copy by the last component and the second copy by the first component (breaking ties by characters identifiers), and merge them to form quintuples.  (It is the copying and sorting step that we do not see how to do in the StreamSort model.)  We sort the set of quintuples by the fourth component, breaking ties by the third component (ignoring characters' identifiers), breaking continued ties by the second component, and breaking continued ties arbitrarily.  (In the first round, all the fourth and second components are 1, so we effectively sort by the third component; notice the third component is the first component's successor in $s$ and the fifth component's predecessor in $s$, taking indices modulo \(n + 1\).)  Finally, we replace the middle triples --- the second, third and fourth components --- with numbers, starting at one and incrementing whenever we find a triple different from the one before.  This process results in another set of triples; if the second components are the numbers 1 through \(n + 1\), we stop; otherwise, we repeat the procedure from the point of copying the triples.

Notice that, at the end of the first round, the first and third components in any triple are two positions apart in $s$, taking indices modulo \(n + 1\).  Also, for any two triples \((s_i, x, s_{i + 2})\) and \((s_j, y, s_{j + 2})\), the comparative relationship between $x$ and $y$ is the same as the lexicographic relationship between $s_{i + 1}$ and $s_{j + 1}$.  At the end of the second round, the first and third components in any triple are four positions apart in $s$ and, for any two triples \((s_i, x, s_{i + 4})\) and \((s_j, y, s_{j + 4})\), the comparative relationship between $x$ and $y$ is the same as the lexicographic relationship between \(s_{i + 3} s_{i + 2} s_{i + 1}\) and \(s_{j + 3} s_{j + 2} s_{j + 1}\). To see why, notice the relationship between $x$ and $y$ depends, in decreasing order of priority, on the relationships between $x_1$ and $y_1$, $s_{i + 2}$ and $s_{j + 2}$, and $x_2$ and $y_2$ in the quintuples
\((s_i, x_2, s_{i + 2}, x_1, s_{i + 4})\) and \((s_j, y_2, s_{j + 2}, y_1, s_{j + 4})\); since these quintuples are formed by joining triples \((s_i, x_2, s_{i + 2})\) and \((s_{i + 2}, x_1, s_{i + 4})\) and \((s_j, y_2, s_{j + 2})\) created during the first round, the comparative relationships between $x_1$ and $y_1$ and $x_2$ and $y_2$ are the same as the lexicographic relationships between $s_{i + 3}$ and $s_{j + 3}$ and $s_{i + 1}$ and $s_{j + 1}$.  After \(O (\log n)\) rounds, when the second components are the numbers 1 through n, they indicate the lexicographic relationships of the prefixes of the third components, so the third components are \BWTs --- in our example in Figure~\ref{fig:BWT Enc}, \(\BWTs = \mathrm{m s \# s p i p i s s i i}\).  (We note that, if we sorted quintuples by the second component, using the third and fourth to break ties, then we would compute the suffix array of $s$.)  We need only \(O (\log n)\) bits of memory for this procedure; each sorting step takes \(O (\log n)\) passes and each other step takes \(O (1)\) passes, so we use \(O (\log^2 n)\) passes altogether.

To invert \BWTs, we again tag each character in \BWTs\ with a unique identifier and form triples.  This time, however, we prepend a ? to each character, then prepend the corresponding character in the stable sort of \BWTs; finally, in the triple whose first component is the special character not in the alphabet, we replace the ? by \(n + 1\).  As for computing \BWTs, we make two copies of the set of triples, sort them and merge them.  This time, for each triple, if the second component is not a ? or the third component is a ?, then we simply delete the third and fourth components; if the second component is a ? and the third component is, then we put the third component minus 1 in the second component, then delete the third and fourth components.  This process results in another set of triples; if the second components are the numbers 1 through \(n + 1\) in some order, we stop; otherwise, we repeat the procedure from the point of copying the triples.

By the definition of the BWT, the $i$th character in the stable sort of \BWTs\ is the predecessor in $s$ of the $i$th character in \BWTs.  Therefore, at the end of the first round, the first and third components in any triple are two positions apart in $s$, taking indices modulo \(n + 1\); also, the triples that have $s_n$ and $s_{n + 1}$ as their first components have $n$ and \(n + 1\) as their second components.  At the end of the second round, the first and third components of any triple are 4 positions apart in $s$ and the triples that have $s_{n - 2}$, $s_{n - 1}$, $s_n$ and $s_{n + 1}$ as their first components have \(n - 2\), \(n - 1\), $n$ and \(n + 1\) as their second components.  After \(O (\log n)\) rounds, when the second components are the numbers from 1 through \(n + 1\) in some order, they indicate the positions in $s$ of the triples' first components.  Sorting by the second component and ignoring the special character, we can recover $s$.  In our example in Figure~\ref{fig:BWT Dec}, the triple that starts `m' then has a 1 as its second component; the triples that start `i' then have 2, 5, 8 and 11; the triples that start `p' then have 9 and 10; and the triples that start `s' then have 3, 4, 6 and 7; thus, sorting them by the second component and outputting the first, we recover \(\mathrm{m i s s i s s i p p i}\).  Again, we need only \(O (\log n)\) bits of memory for this procedure; each sorting step takes \(O (\log n)\) passes and each other step takes \(O (1)\) passes, so we use \(O (\log^2 n)\) passes altogether. \qed
\end{proof}

Grohe and Schweikardt showed that, given a sequence of $n$ numbers \(x_1, \ldots, x_n\), each of \(2 \log n\) bits, we cannot generally sort them using \(o (\log n)\) passes, \(O (n^{1 - \epsilon})\) bits of memory and \(O (1)\) tapes, for any positive constant $\epsilon$.  We can easily obtain the same lower bound for the BWT, via the following reduction from sorting: given \(x_1, \ldots, x_n\), using one pass, \(O (\log n)\) bits of memory and two tapes, for \(1 \leq i \leq n\) and \(1 \leq j \leq 2 \log n\), we replace the $j$th bit \(x_i [j]\) of $x_i$ by \(x_i [j]\ 2\ x_i\ i\ j\), writing 2 as a single character, $x_i$ in \(2 \log n\) bits, $i$ in \(\log n\) bits and $j$ in \(\log \log n + 1\) bits; notice the resulting string is of length \(n (3 \log n + \log \log n + 2)\).  (For simplicity, we now consider each character's context as starting at its successor and extend forwards.)  The only characters followed by 2s in this string are the bits at the beginning of replacement phrases so, if we perform the BWT of it, the last \(2 n \log n\) characters of the transformed string are the bits of \(x_1, \ldots, x_n\); moreover, since the lexicographic order of equal-length binary strings is the same as their numeric order, the \(x_i [j]\)s will be arranged by $x_i$s, with ties broken by the $i$s (so if \(x_i = x_{i'}\) with \(i < i'\), then every \(x_i [j]\) comes before every \(x_{i'} [j']\)), and further ties broken by the $j$s; thus, the last \(2 n \log n\) bits of the transformed string are \(x_1, \ldots, x_n\) in sorted order.

\bibliographystyle{plain}
\bibliography{bounds}

\end{document}